\documentclass[11pt,a4paper]{article}
\usepackage{fullpage}
\usepackage{hyperref}
\usepackage[ruled]{algorithm2e} 
\usepackage{amsmath}
\usepackage{amsfonts}
\usepackage{amsthm}
\usepackage{tikz}
\usepackage{caption}
\usepackage{subcaption}

\usetikzlibrary{quotes,angles}
\usetikzlibrary{calc}
\usetikzlibrary{decorations.pathreplacing}

\newcommand{\N}{\mathcal{N}}
\newcommand{\E}{\mathbb{E}}
\newcommand{\cut}{\text{CUT}}
\newcommand{\bigO}{\mathcal{O}}
\newcommand{\wgt}{w}
\newcommand{\vb}{\mathbf{v}}

\SetAlFnt{\small}
\SetAlCapFnt{\small}
\SetAlCapNameFnt{\small}
\SetAlCapHSkip{0pt}
\IncMargin{-\parindent}

\usepackage{algorithmic}

\newcommand*\samethanks[1][\value{footnote}]{\footnotemark[#1]}

\newcommand{\Sout}{S_{\text{out}}}
\newcommand{\Umax}{U_{\text{max}}}
\newcommand{\Umin}{U_{\text{min}}}
\newcommand{\game}{\mathcal{G}}
\newcommand{\R}{\mathbb{R}}
\newcommand{\ROUND}{\text{ROUND}}
\newcommand{\SDP}{\text{SDP}}
\newcommand{\GM}{\text{GLOBAL}}
\newcommand{\LM}{\text{LOCAL}}

\newcommand{\xor}{\text{XOR}}

\newtheorem{theorem}{Theorem}[section]
\newtheorem{lemma}[theorem]{Lemma}
\newtheorem{claim}[theorem]{Claim}

\allowdisplaybreaks

\title{\bf Computing better approximate pure Nash equilibria\\ in cut games
via semidefinite programming}

\author{Ioannis Caragiannis\thanks{Department of Computer Science, Aarhus University, {\AA}bogade 34, 8200 Aarhus N, Denmark. Email:~\url{{iannis,zhile}@cs.au.dk}} \and Zhile Jiang\samethanks}

\date{}

\begin{document}

\maketitle

\begin{abstract}
Cut games are among the most fundamental strategic games in algorithmic game theory. It is well-known that computing an exact pure Nash equilibrium in these games is PLS-hard, so research has focused on computing approximate equilibria. We present a polynomial-time algorithm that computes $2.7371$-approximate pure Nash equilibria in cut games. This is the first improvement to the previously best-known bound of $3$, due to the work of Bhalgat, Chakraborty, and Khanna from EC 2010. Our algorithm is based on a general recipe proposed by Caragiannis, Fanelli, Gravin, and Skopalik from FOCS 2011 and applied on several potential games since then. The first novelty of our work is the introduction of a phase that can identify subsets of players who can simultaneously improve their utilities considerably. This is done via semidefinite programming and randomized rounding. In particular, a negative objective value to the semidefinite program guarantees that no such considerable improvement is possible for a given set of players. Otherwise, randomized rounding of the SDP solution is used to identify a set of players who can simultaneously improve their strategies considerably and allows the algorithm to make progress. The way rounding is performed is another important novelty of our work. Here, we exploit an idea that dates back to a paper by Feige and Goemans from 1995, but we take it to an extreme that has not been analyzed before. \end{abstract}

\section{Introduction}
Understanding the computational aspects of equilibria is a key goal of algorithmic game theory. In this direction, we consider the fundamental class of cut games. In a cut game, each node of an edge-weighted graph is controlled by a distinct player. The players aim to build a cut of the graph in a non-cooperative way. Each player has two strategies, to put the node she controls in one of the two sides of the cut, and aims to maximize her utility, defined as the total weight of edges in the cut that are incident to the node the player controls. A pure Nash equilibrium (or, simply, an equilibrium) is a state of the game in which no player has any incentive to unilaterally change her strategy in order to improve her utility.

Cut games are potential games. The total weight of the edges in the cut is a potential function, with the following remarkable property: for every two states that differ in the strategy of a single player, the difference in the value of the potential function is equal to the difference in the utility of the deviating player. Thus, a sequence of best-response moves, consisting of players who change their strategies and improve their utilities, is guaranteed to lead to a state that locally maximizes the potential function and is, thus, an equilibrium. So, finding an equilibrium is equivalent to computing a locally-maximum cut in a graph, i.e., a cut that cannot be improved by moving a single node from one side of the cut to the other. 

Unfortunately, computing a locally-maximum cut in a graph is a PLS-hard problem~\cite{SY91}. In light of this, approximate pure Nash equilibria seem to be a reasonable compromise for cut games; in general, a $\rho$-approximate pure Nash equilibrium is a state of a strategic game in which no player deviation can increase her utility by a factor of more than $\rho$. Among other results, Bhalgat et al.~\cite{BCK10} present a polynomial-time algorithm for computing a $3$-approximate equilibrium in cut games.

At first glance, two algorithmic ideas seem relevant for tackling the challenge of computing approximate pure Nash equilibria. The first one is to follow a sequence of deviations by players who improve their utility by a factor of more than $\rho$. The existence of a potential function guarantees that this process will eventually converge to a $\rho$-approximate equilibrium. Unfortunately, such sequence can be exponentially long, as Bhalgat et al.~\cite{BCK10} have shown specifically for cut games (and for small approximation factors). The second one is to exploit an approximation algorithm for the problem of maximizing the potential function. For cut games, this could involve the celebrated algorithm of Goemans and Williamson~\cite{GW95} for MAX-CUT or, more generally, excellent approximations of local maxima using the techniques of Orlin et al.~\cite{OPS04}. Unfortunately, approximations of the potential function and approximate equilibria are unrelated notions. So, the algorithm in~\cite{BCK10} exploits the structure of cut games to define restricted subgames in which sequences of player deviations are applied separately. This approach leads to $3$-approximate equilibria in cut games and is also applicable to the broader class of constraint satisfaction games.

\paragraph{Our results and techniques.} 
In this paper, we present the first improvement to the result of Bhalgat et al.~\cite{BCK10}, showing how to compute efficiently a $2.7371$-approximate equilibrium in a cut game. Our result follows a general algorithmic recipe proposed by Caragiannis et al.~\cite{CFGS11} for congestion games. Applications of this recipe to constraint satisfaction games~\cite{CFG17} has improved many of the approximation bounds presented in~\cite{BCK10} but not the one for cut games (even though it has rediscovered a $3$-approximation). Unlike~\cite{BCK10}, the algorithms in~\cite{CFGS11,CFG17} identify polynomially-long sequences of coordinated player deviations that eventually lead to an approximate equilibrium. Our improvement is possible by introducing a {\em global move step} in which a set of players change their strategies simultaneously. Such a set of players is identified by solving a semidefinite program and rounding its solution. To the best of our knowledge, this is the first application of semidefinite programming to the computation of approximate pure Nash equilibria. Essentially, our paper blends the two algorithmic approaches mentioned in the previous paragraph for the first time.

The general structure of the algorithm in~\cite{CFG17} is as follows. Players are classified into blocks, so that players within the same block are polynomially-related in terms of their maximum utility. Then, a set of phases is executed. In each phase, the players in two consecutive blocks are allowed to move. The players in the (heavy) block of higher maximum utility play (roughly) $3$-moves, i.e., they change their strategy whenever they can improve their utility by a factor of more than $3$. The players in the (light) block of lower maximum utility play $(1+\varepsilon)$-moves (for some tiny positive $\varepsilon$). At the end of the phase, the strategies of the players in the heavy block are irrevocably decided. The crucial property of the algorithm is that the light players of a phase are the heavy ones in the next one. Due to the moves of the light players, the state reached at the end of a phase has {\em low stretch} (of at most $3$) for the light players, meaning that there is no subset of them who could simultaneously change their strategy and improve their contribution to the potential function by a factor of more than $3$. This guarantees that the number of $3$-moves by the heavy players in the next phase will be polynomially-bounded and, furthermore, the effect that these moves have to players whose strategies were irrevocably decided in previous phases is minimal. So, all players stay close to a $3$-approximate equilibrium at the end of the whole process.

Our algorithm follows the same general structure, even though it coordinates more carefully the moves of the players in different blocks within a phase. The improvement comes from a different implementation of the moves of the light players within a phase. Here, we focus explicitly on the problem of detecting whether there exists a set of light players who can improve their contribution to the potential by a factor of $\rho$ by changing their strategies simultaneously. Essentially, the algorithm of \cite{CFG17} solves this problem for $\rho=3$ using local search (the analog of best-response moves by the light players). Instead, we obtain the better approximation of $2.7371$ by formulating this problem as a semidefinite program. Non-positive objective values of the semidefinite program guarantee that the current cut has low stretch for the light players. Positive objective values guarantee that by rounding the SDP solution appropriately, we obtain a global move which improves the potential considerably and makes progress, similar to the progress an $(1+\varepsilon)$-move by a single light player makes in the algorithm of \cite{CFG17}.

Rounding the SDP solution is very challenging. The SDP objective has a mixture of XOR-type terms, requiring that the two vectors corresponding to the endpoints of an edge that is not in the current cut are opposite, and OR-type terms requiring that one of the endpoints of an edge in the current cut coincides with a special vector the SDP is using. Even though the XOR-type terms have non-negative coefficients and standard hyperplane rounding could yield excellent approximations for the particular terms (like in~\cite{GW95}), the OR-type terms have negative coefficients, making standard hyperplane rounding disastrous for them. So, we instead use an idea that originates from Feige and Goemans~\cite{FG95}, which has inspired much follow-up work on approximating MAX2-SAT and related problems, but we consider it in an extreme that has not been considered before. In particular, the vectors in the SDP solution are first rotated in the $2$-dimensional plane they define with the special vector and, then, hyperplane rounding is performed. We use the rotation function $f(\theta)=\frac{\pi}{2}(1-\cos{\theta})$, meaning that a vector at an angle of $\theta$ from the special vector is relocated at an angle $f(\theta)$ from it. This gives a rather poor $1/2$-approximation to the XOR-type terms but approximates well the OR-type constraints.

The approximation ratio we obtain is the solution of a quadratic equation; the exact value is $\frac{1+2\sqrt{13}}{3}\approx 2.7371$. This is extremely rare for the approximation factors obtained from SDPs in combinatorial optimization problems. To prove the properties of rounding, we need to prove two inequalities for triplets of $3$-dimensional vectors, which turn out to depend on two variables only: the angles the special vector forms with the two other vectors involved in a term of the SDP objective. As it is often the case with the analysis of SDP-based approximation algorithms (e.g., in~\cite{LLZ02}), we have verified one of these inequalities by extensive numerical computations. These imply that the worst-case bound is obtained for a restriction of the parameters to more tractable cases, for which we present a formal analysis.

\paragraph{Related work.}
Even though research in algorithmic game theory gained popularity after the seminal papers by Koutsoupias and Papadimitriou~\cite{KP09}, Nisan and Ronan~\cite{NR01}, and Roughgarden and Tardos~\cite{RT02} (the conference versions appeared in 1999, 1999, and 2000, respectively), research in computational aspects of pure Nash equilibria had started earlier, with the introduction of the class PLS ---standing for polynomial local search--- by Johnson, Papadimitriou, and Yannakakis~\cite{JPY88} and subsequent (mainly negative) results on computing locally-optimum solutions for combinatorial optimization problems. Mostly related to the topic of the current paper is the paper by Schaeffer and Yannakakis~\cite{SY91}. Expressed in game-theoretic terms, they proved that computing a pure Nash equilibrium in a cut game is a PLS-hard problem.

The connection of cut games and local search should not come as a surprise. Cut games are potential games~\cite{MS96}. They admit a potential function defined over the states of the game so that for any two states differing in the strategy of a single player, the difference in the potential values between the two states and the difference in the utility of the deviating player has the same sign. The class of potential games also includes congestion games~\cite{FPT04,MS96,R73}, network design games~\cite{ADKTWR08}, constraint satisfaction games~\cite{BCK10,CFG17}, market sharing games~\cite{GLMT06}, and more. Most of these games have the additional property that the two differences mentioned above not only have the same sign but they are actually equal. Such potential functions are usually called exact~\cite{FPT04}. For cut games, the total weight of edges that cross the cut is an exact potential function.

Early attempts on computing equilibria in potential games have aimed to follow a best-response dynamics, i.e., a sequence of player moves, in which the utility of the deviating player increases by at least a $1+\varepsilon$ factor. This has been proved useful in a few cases ---e.g., in restricted cases of congestion games~\cite{ARV08,BCK09,CS11}--- to compute $(1+\varepsilon)$-approximate equilibria. Such approaches work nicely if the players are similar, in the sense that their moves result in comparable utility increases. Simple arguments, which relate the increase in the potential function with the minimum increase in the utility of a deviating player can be used to bound the running time (i.e., the length of the best-response dynamics to the approximate equilibrium) in terms of the game size and $1/\varepsilon$. However, players rarely have such similarities. In a cut game, we may have players controlling nodes that have incident edges of high weight and others which have only incident edges of negligible weight. Then, such approaches cannot compute approximate equilibria in polynomial time, even though they can be used to compute states of high social welfare, e.g., see~\cite{AAEMS08,BBM09,BP16,CMS12}.

The most successful recipe for computing approximate equilibria in potential games has been proposed by Caragiannis et al.~\cite{CFGS11}. Their approach has yielded $2$-approximate equilibria for linear congestion games, and $d^{\bigO(d)}$-approximate equilibria for congestion games with polynomial latency functions of degree $d$. These results have been improved in~\cite{FGS14,VS20}. More importantly, the recipe of~\cite{CFGS11} has been applied to other potential games such as constraint satisfaction~\cite{CFG17} and cost sharing games~\cite{GKK20}, as well as to non-potential games such as weighted congestion games in~\cite{CF21,CFGS15,GNS22}. These last results have been possible by either approximating the original game by a potential game in which the algorithm is applied, or by introducing an approximate potential function to keep track of best-response moves in the original games. In all applications of the recipe of~\cite{CFGS11}, the algorithm identifies a sequence of deviations by the players, either in the original game or in its approximation. The current paper is the first one to extend the recipe with global moves (of simultaneous player deviations).

Our global moves are implemented via semidefinite programming and randomized rounding. Starting with the celebrated paper of Goemans and Williamson~\cite{GW95}, these tools have been proved very useful in approximating combinatorial optimization problems such as MAX-CUT (the problem of computing a cut of maximum total weight) and MAX-2SAT (the problem of computing a Boolean assignment that maximizes the total weight of satisfied clauses in a formula in conjunctive normal form), among many others. The classical semidefinite program for MAX-CUT uses a unit vector $v_j$ for each graph node $j$, to represent the side of the cut in which to node is put. The objective is just the sum of the terms $\frac{1}{2}(1- v_j\cdot v_k)$ multiplied by the weight of edge $(j,k)$ over all edges in the input graph. Notice that the quantity $\frac{1}{2}(1- v_j\cdot v_k)$ is a XOR-type term, taking the value $1$ if the two vectors are opposite and $0$ if the two vectors coincide (corresponding to nodes in different sides or at the same side of the cut, respectively). The solution of the SDP, which may inevitably consist of vectors in many different directions, is rounded by picking a random hyperplane crossing the vectors' origin, and putting the nodes at the left or the right of the cut, according to side of the hyperplane the corresponding unit vectors are lying. In this way, Goemans and Williamson~\cite{GW95} improved the trivial $1/2$-approximation for MAX-CUT to $0.8785$. 

Modeling OR-type terms (like those required by MAX-2SAT) requires an extra variable, a reference (or special) vector $\widehat{v}$. Then, using the vector $v_j$ as (an approximation to) the value of a Boolean variable, the OR-type term $\frac{1}{4}\cdot (3+v_j\cdot \widehat{v}+v_k\cdot \widehat{v}-v_j\cdot v_k)$ takes the value $1$ if at least one of the vectors $v_j$ and $v_k$ coincide with the reference vector $\widehat{v}$. Feige and Goemans~\cite{FG95} considered improved the approximation bound of~\cite{GW95} for MAX-2SAT by exploiting the existence of the reference vector $\widehat{v}$ to rotate all vectors appropriately before rounding them using a random hyperplane. Rotation relocates each vector $v_j$ forming an angle $\theta$ with the reference vector at a new position in the $2$-dimensional plane defined by $v_j$ and $\widehat{v}$, at angle $f(\theta)$ from $\widehat{v}$. The rotation function $f$ they use is a linear combination between the no-rotation function $f(\theta)=\theta$ and the rotation function $f(\theta)=\frac{\pi}{2}\cdot (1-\cos{\theta})$ that we also consider here. An extensive discussion of this approach with additional details is given in~\cite{Z00}. Lewin et al.~\cite{LLZ02} present the best known bound of $0.9401$ for MAX-2SAT; in their paper, they discuss several alternatives for rounding the solution of SDPs for MAX-2SAT. Finally, we remark that even though SDPs and rounding have been used to handle negative coefficients successfully, e.g., in approximating the cut-norm of a matrix~\cite{AN06} or in maximizing quadratic programs~\cite{CW04}, those formulation are very far from ours and the corresponding techniques do not seem to be applicable in our case.

\paragraph{Roadmap.} The rest of the paper is structured as follows. We begin with preliminary definitions and notation in Section~\ref{sec:prelim}. Our algorithm is presented in Section~\ref{sec:alg}, which concludes with the statement of our main result (Theorem~\ref{thm:main}). Then, Section~\ref{sec:global-moves} presents the properties of the global-moves routine of our algorithm. We put everything together and prove Theorem~\ref{thm:main} in Section~\ref{sec:proof}, while the proof of Lemma~\ref{lem:xor}, which states the properties of randomized rounding is deferred to Section~\ref{sec:xor-proofs}. We conclude in Section~\ref{sec:open}.

\section{Preliminaries}\label{sec:prelim}
A {\em cut game} is defined by an edge-weighted complete $n$-node graph $G=(V,E)$. Each edge $e$ of the graph has a non-negative weight $w_e$. We overload this notation and write $w_{jk}$ to denote the weight of the edge $(j,k)$ and $\wgt(A)=\sum_{e\in A}{w_e}$ to denote the total weight of the edges in set $A$. There are $n$ players, each corresponding to a distinct node of $G$. The players aim to collectively but non-cooperatively build a {\em cut} of the graph, i.e., a partition of its nodes into two sets. Player $j\in \N$ has two different strategies: putting her node $i$ at the {\em left} or at the {\em right} side of the cut. A {\em state} of the game is represented by an $n$-vector $S=(s_1,s_2, ..., s_n)$, indicating that player $j$ uses strategy $s_j$. Given a state $S$, we denote by $\cut(S)$ the set of edges whose endpoints correspond to players that select different strategies in $S$. For a subset of players $R\subseteq \N$, we denote by $\cut_R(S)$ the subset of $\cut(S)$ that consists of edges with at least one endpoint corresponding to a player in $R$. With some abuse of notation, we simplify $\cut_{\{j\}}(S)$ to $\cut_j(S)$. The {\em utility} of player $j$ is the total weight of the edges incident to her node, i.e., $u_j(S)=\wgt(\cut_j(S))$. We also denote by $E_R$ the set of edges that are incident to at least one player in $R$ and simplify $E_{\{j\}}$ to $E_j$. We denote by $U_j$ the {\em maximum utility} that player $j$ may have, i.e., $U_j=\wgt(E_j)$.

We use a standard game-theoretic notation. Given a state $S=(s_1, s_2, ..., s_n)$ and a player $j\in \N$, we denote by $(S_{-j},s'_j)$ the state obtained by $S$ when player $j$ unilaterally changes her strategy from $s_j$ to its complement $s'_j$. This is an {\em improvement move} (or, simply, a {\em move}) for player $j$ if her utility increases, i.e., if $u_j(S_{-j},s'_j)>u_j(S)$. We call it a {\em $\rho$-move} when the utility of player $j$ increases by more than a factor of $\rho>1$, i.e., $u_j(S_{-j},s'_j)>\rho\cdot u_j(S)$. A state $S$ is a {\em pure Nash equilibrium} (or, simply, an {\em equilibrium}) if no player has a move to make. Similarly, $S$ is a {\em $\rho$-approximate (pure Nash) equilibrium} if no player has any $\rho$-move to make.

Given a state $S$, we denote by $\Phi(S)$ the total weight of the edges in $\cut(S)$, i.e., $\Phi(S)=\wgt(\cut(S))$. The function $\Phi$ is a {\em potential function}, having the remarkable property that for every two states $S$ and $(S_{-j},s'_j)$ that differ only in the strategy of player $j$, the difference in the potential function is equal to the difference of the utility of player $j$, i.e., $\Phi(S)-\Phi(S_{-j},s'_j)=u_j(S)-u_j(S_{-j},s'_j)$.

Following~\cite{CFG17}, we will often consider sequences of moves in which only players from a certain subset $R\subseteq \N$ are moving. These can be thought of as moves in a {\em subgame} defined on graph $G$ among the players in $R$ only, in which the nodes that are not controlled by players in $R$ are {\em frozen} to a fixed placement in one of the sides of the cut. We represent states of a subgame in the same way we do for the original game using $n$-vectors of strategies, including both the strategies of the players in $R$ and the frozen strategies of the players in $\N\setminus R$. 

The next two claims follow easily by our definitions and are used repeatedly in our proofs. Claim~\ref{claim:subgame-potential} shows that the subgames of a cut game are also potential games and Claim~\ref{claim:potential-increase-by-p-move} gives basic facts about the utility of players.

\begin{claim}\label{claim:subgame-potential}
Consider a cut game $\game$ with a set $\N$ of players. For every subgame of $\game$ defined among a subset $R\subseteq \N$ of players, the function $\Phi_R$, defined as $\Phi_R(S)=\wgt(\cut_R(S))$ for each state $S$ of the subgame, is an exact potential function.
\end{claim}

\begin{claim}\label{claim:potential-increase-by-p-move}
Consider a cut game $\game$ with a set $\N$ of players, and let $j\in \N$ be a player and $S$ a state of $\game$. If player $j$ performs a $p$-move, her utility and the potential $\Phi_R$ for every subgame involving her (i.e., $j\in R\subseteq \N$) increase by more than $\frac{p-1}{p+1}\cdot U_j$. If player $j$ does not have any $p$-move to make at state $S$, then $u_j(S)\geq \frac{1}{p+1}\cdot U_j$.
\end{claim}

Given two states $S$ and $T$ of game $\game$ and a possibly empty subset of players $R\subseteq \N$, we denote by $(S_{-R},T_R)$ the state in which each player of set $R$ uses her strategy in state $T$ and each player in set $\N\setminus R$ uses her strategy in state $S$. In particular, we denote by $(S_{-R},S'_R)$ the state in which each player of set $R$ uses the complement of her strategy in state $S$ and each player of set $\N\setminus R$ uses her strategy in $S$. We refer to the ratio $\frac{\Phi_R(S_{-R},S'_R)}{\Phi_R(S)}$ as the {\em stretch} of state $S$ for the subgame defined by the subset of players $R\subseteq \N$. 

\section{The algorithm}\label{sec:alg}
Our algorithm (see its pseudocode as Algorithm~\ref{alg:main}) has a simple general structure. It uses several constants: a tolerance parameter $\varepsilon\in [0,1/4]$ and the constants $\rho=\frac{1+2\sqrt{13}}{3}\approx 2.7371$, $\sigma=\rho+\varepsilon/3$, and $\tau=\rho+2\varepsilon/3$. The algorithm takes as input a cut game $\game$ with $n$ players and returns a state $\Sout$. As we will show, the algorithm terminates after polynomially many steps in terms of $n$ and $1/\varepsilon$ and the state $\Sout$ is a $(\rho+\varepsilon)$-approximate pure Nash equilibrium, with high probability.

\begin{algorithm}[ht]
 		{ 

 			{\bf Input:} A cut game $\game$ with a set $\N$ of $n$ players 
 			
 			{\bf Output:} A state $\Sout$ of $\game$
 			\caption{The algorithm for computing approximate pure Nash equilibria in cut games} \label{alg:main}
 			\begin{algorithmic}[1]
 				\STATE $\Delta\leftarrow 480n/\varepsilon^2$;\label{alg:line:1}
 				\STATE Set $\Umin\leftarrow \min_{j\in \N}{U_j}$, $\Umax\leftarrow \max_{j\in \N}{U_j}$, and $m\leftarrow 1 + \left\lfloor \log_\Delta{\left(\Umax/\Umin\right)}\right\rfloor$ ;\label{alg:line:2}
 				\STATE (Implicitly) partition players into block $B_1$, $B_2$, ..., $B_m$, such that $j\in B_i$ implies that $U_j\in \left(\Umax \Delta^{-i}, \Umax\Delta^{1-i}\right]$;\label{alg:line:3}
 				\STATE Set $S\leftarrow$ an arbitrary initial state;\label{alg:line:3+}
 				\STATE Update $S\leftarrow \GM(\game,S,B_1)$;\label{alg:line:4}
 				\FOR{phase $i\leftarrow 1$ to $m-1$ such that $B_i\not=\emptyset$}
 				\REPEAT
 				\STATE Update $(S,ch1)\leftarrow \LM(\game,S,B_i,\tau)$;
 				\STATE Update $(S,ch2)\leftarrow \GM(\game,S,B_{i+1})$;
 				\UNTIL{$ch1=\FALSE$ \AND $ch2=\FALSE$}
 				\ENDFOR
 				\STATE Update $S\leftarrow \LM(\game,S,B_m,\tau)$;\label{alg:line:10}
 				\RETURN $\Sout\leftarrow S$;
   		\end{algorithmic}
 		}
 	\end{algorithm}

The algorithm initializes parameter $\Delta$ to be polynomially related to $n$ and $1/\varepsilon$ (line \ref{alg:line:1}). Then, in lines \ref{alg:line:2} and \ref{alg:line:3}, it implicitly partitions the players into blocks $B_1$, $B_2$, ..., $B_m$ according to their maximum utility. Denoting by $U_{\max}$ the maximum utility among all players, block $B_i$ consists of the players with maximum utility in the interval $(U_{\max}\Delta^{-i},U_{\max}\Delta^{1-i}]$. By the definition of $\Delta$, the players in the same block have polynomially related maximum utilities.

The lines \ref{alg:line:3+}-\ref{alg:line:10} are the main body of the algorithm. The algorithm utilizes two routines, namely $\LM$() and $\GM$(), implementing a set of local and global moves among by players in a particular block. The computation is divided into phases. Starting with an arbitrary initial state (in line~\ref{alg:line:3+}), phase 0 (line \ref{alg:line:4}) executes global moves for the players in block $B_1$. Then, for $i=1, ..., m-1$, phase $i$ repeatedly executes local moves for the players in block $B_i$ and global moves for the players in block $B_{i+1}$ until the execution of these routines does not produces new states. Finally, the algorithm runs phase $m$, which executes local moves for the players in block $B_m$. The state produced after the execution of phase $m$ is returned as the output state $\Sout$.

The routine $\LM$() (see Algorithm~\ref{alg:local-moves}) takes as input the cut game $\game$, a current state $S$, a subset of players $B$, and a parameter $p>1$; on exit, it returns an output state and a flag denoting whether the output state is different than the input one. $\LM$() repeatedly performs $p$-moves as long as there is a player in set $B$ that has such a $p$-move to make.

\begin{algorithm}[ht]
 		{ 

 			{\bf Input:} A cut game $\game$ with a set $\N$ of $n$ players, a state $S$, a subset $B\subseteq \N$ of players, and a parameter $p>1$
 			
 			{\bf Output:} A state of $\game$ and a Boolean flag
 			\caption{The routine \LM()}	\label{alg:local-moves}
 			\begin{algorithmic}[1]
 				
 				\STATE Set $change\leftarrow \FALSE$;
 				\WHILE{there is a player $j\in B$ that satisfies $u_j(S_{-j},s'_j)>p\cdot u_j(S)$}
 				\STATE Update $S\leftarrow (S_{-j},s'_j)$;
 				\STATE $change\leftarrow \TRUE$;
 				\ENDWHILE
 				\RETURN $(S,change)$;
   		\end{algorithmic}
 		}
 	\end{algorithm}
 	
The routine $\GM$() is the main novelty of our algorithm compared to the algorithm in~\cite{CFG17}. It calls the routines $\SDP$() and $\ROUND$(). When executed on input a cut game $\game$ with a set $\N$ of $n$ players, state $S$, and the subset of players $B\subseteq \N$, the routine $\SDP$() solves the following semidefinite program:
\begin{eqnarray*}
&\mbox{maximize} & Z=\frac{1}{2}\cdot \sum_{(j,k)\in E\setminus \cut_{B}(S)}{w_{jk}\cdot (1-v_j\cdot v_k)}\\ & &-\frac{1}{4}\cdot \sum_{(j,k)\in \cut_{B}(S)}{w_{jk}\cdot \left(3\sigma-1+(\sigma-1)\cdot\widehat{v} \cdot v_j+(\sigma-1)\cdot \widehat{v}\cdot v_k-(\sigma+1) \cdot v_j\cdot v_k\right)}\\
&\mbox{subject to} &\sum_{(j,k)\in \cut_{B}(S)}{w_{jk}\cdot (3+v_j\cdot \widehat{v}+v_k\cdot \widehat{v}-v_j\cdot v_k)} \geq \min_{j\in B}{U_j}\\
& & \widehat{v}\cdot \widehat{v}=1, \widehat{v}\in \R^n\\
& & v_j\cdot v_j=1, v_j\in \R^n, \forall j\in B\\
& & v_j = -\widehat{v}, \forall j\in \N\setminus B
\end{eqnarray*}
The output of $\SDP$() is the matrix $\vb$ containing as columns the unit vectors $v_j$ for $j\in \N$ and $\widehat{v}$, and the objective value $Z$.

At a high level, the semidefinite program aims to identify a set of players $R\subseteq B$ so that the stretch of the current state for the subgame defined by the players in $R$ is at least $\sigma$. The special vector $\widehat{v}$ is used as a basis for the orientation of the remaining vectors. The vectors $v_j$ for $j\in \N\setminus B$ are set to $-\widehat{v}$, indicating that the players in $\N\setminus B$ are not considered for inclusion in the set $R$. For $j\in B$, vector $v_j$ indicates whether player $j$ belongs to set $R$ (when $v_j=\widehat{v}$) or not (when $v_j=-\widehat{v}$). Of course, the vectors in a solution of the semidefinite program can be much different than these two extreme ones. As we will see, if the semidefinite program has non-positive objective value, no set $R$ of players defining a subgame of high stretch at the current state exists, and this provides a useful condition for our algorithm. Otherwise, a subsequent appropriately performed (randomized) rounding of the SDP solution will yield a set of simultaneous (global) moves by players that increase the potential considerably. 

This is done by the routine $\ROUND$(), which takes as input a cut game $\game$ with a set $\N$ of $n$ players, a state $S$, a subset of players $B\subseteq \N$, and a matrix $\vb$ having vectors $v_j$ for each player $j\in B$ and the special vector $\widehat{v}$ as columns. $\ROUND$() performs $\frac{23040n^2}{\varepsilon^3}\cdot 
\ln{\left(\frac{414720n^4}{\varepsilon^4}\right)}$ iterations. In each of them, it decides a set $R\subseteq \N$ of players by rounding the set of unit vectors $\vb$. Rounding is performed as follows. First, for $j\in B$, each vector $v_j$ is rotated to produce the unit vector $v'_j$ as follows: $v'_j$ is the unit vector in $\R^n$ that lies on the plane defined by vectors $v_j$ and $\widehat{v}$ and is at angle $\frac{\pi}{2}\cdot (1-v_j\cdot \widehat{v})$ from vector $\widehat{v}$ and at angle $|\arccos{(v_j\cdot \widehat{v})}-\frac{\pi}{2}\cdot (1-v_j\cdot \widehat{v})|$ from vector $v_j$. Then, a random unit vector $r\in \R^n$ is picked. For $j\in B$, player $j$ is included in set $R$ if the inner products $v'_j\cdot r$ and $\widehat{v}\cdot r$ have the same sign. Informally, we can view vector $r$ as the unit normal vector of a random hyperplane; then, player $j$ is included in set $R$ if the rotated vector $v'_j$ and vector $\widehat{v}$ are on the same side of the hyperplane. $\ROUND$() keeps the set $R$ of players in the rounding iteration that produces the highest potential increase $\Phi_R(S_{-R},S'_R)-\Phi_R(S)$ and returns state $(S_{-R},S'_R)$ as output.

Now, the routine $\GM$() (see Algorithm~\ref{alg:global-moves}) takes as input the cut game $\game$, a state $S$, and a subset of players $B$; on exit, it returns an output state and a flag indicating whether the output state is different than the input one. $\GM$() repeatedly runs the routing $\SDP$(), as long as the objective value of the demidefinite program it solves is positive. Then, it executes the routine $\ROUND$() to round the solution of the semidefinite program and decide, in this way, the set of players who will change their strategy simultaneously. Essentially, this implements a global move. Prior to executing $\SDP$(), the routine $\LM$() is executed to make sure that the input state to $\SDP$() is an approximate equilibrium for the players in $B$.

\begin{algorithm}[ht]
 		{ 

 			{\bf Input:} A cut game $\game$ with a set $\N$ of $n$ players, a state $S$, and subset $B\subseteq \N$ of players
 			
 			{\bf Output:} A state of $\game$ and a Boolean flag
 			\caption{The routine \GM()}\label{alg:global-moves}	
 			\begin{algorithmic}[1]
 				
 				\STATE Set $change\leftarrow \FALSE$;
 				\STATE Update $(S,change) \leftarrow \LM(\game,S,B,\tau)$
 				\STATE Set $(\vb,Z)\leftarrow \SDP(\game,S,B)$;
 				\WHILE{$Z>0$}
 				\STATE Update $S\leftarrow \ROUND(\game,S,B,\vb)$;
 				\STATE Update $change\leftarrow\TRUE$;
 				 \STATE Update $S \leftarrow \LM(\game,S,B,\tau)$;
 				 \STATE Set $(\vb,Z)\leftarrow \SDP(\game,S,B)$;
 				\ENDWHILE
 				\RETURN $(S,change)$;
   		\end{algorithmic}
 		}
 	\end{algorithm}

This completes the detailed description of our algorithm. We will prove the following statement.

\begin{theorem}\label{thm:main}
Let $\varepsilon\in (0,1/4]$. On input a cut game $\game$ with $n$ players, Algorithm~\ref{alg:main} computes a $(\rho+\varepsilon)$-approximate pure Nash equilibrium $\Sout$ in time $\text{poly}(n,1/\varepsilon)$ with probability at least $1-1/n$, where $\rho=\frac{1+2\sqrt{13}}{3}\approx 2.7371$. 
\end{theorem}

\section{Properties of global moves}\label{sec:global-moves}
We devote this section to proving two important properties of routines $\SDP$() and $\ROUND$(). In our analysis, we denote the upper boundary of block $B_i$ by $W_i$ and by $W_{m+1}$ the lower boundary of block $B_m$, i.e., $W_i=U_{\max}\Delta^{1-i}$ for $i=1, 2, ..., m+1$. So, the players of block $B_i$ are those with maximum utility $U_j\in (W_{i+1},W_i]$.

The next lemma guarantees that, whenever routine $\GM$() terminates, the state returned has stretch at most $\sigma$ for any subgame defined by subsets of players in block $B_{i+1}$. 

\begin{lemma}\label{lem:sdp-stretch}
Consider the execution of routine $\SDP()$ during phase $i\geq 0$, i.e., $\SDP$() takes as input the cut game $\game$, the block $B_{i+1}$ as the set of players $B$, and a state $S$ that is a $\tau$-approximate pure Nash equilibrium for the players in $B_{i+1}$. If the objective value $Z$ of the semidefinite program that $\SDP$() returns is non-positive, then for every set of players $R\subseteq B_{i+1}$, it holds that $\sigma\cdot \Phi_R(S) \geq \Phi_R(S_{-R},S'_R)$.
\end{lemma}

\begin{proof}
Assume otherwise and let $R$ be a non-empty subset of $B_{i+1}$ such that 
\begin{align}\label{eq:assumption}
\sigma\cdot \Phi_R(S) &< \Phi_R(S_{-R},S'_R).
\end{align}
Let $\widehat{v}$ be an arbitrary unit vector and set $v_j=\widehat{v}$ for $j\in R$ and $v_j=-\widehat{v}$ for $j\in \N\setminus R$. Since $R\not=\emptyset$, there is at least one player $j^*\in R$ such that $v_{j^*}=\widehat{v}$. Then, $3+v_{j^*}\cdot \widehat{v}+v_k\cdot \widehat{v}-v_{j^*}\cdot v_k=4$ for every $k\in \N\setminus \{j^*\}$ and 
\begin{align*}
    \sum_{(j,k)\in \cut_{B_{i+1}}(S)}{w_{jk}\cdot (3+v_j\cdot \widehat{v}+v_k\cdot \widehat{v}-v_j\cdot v_k)}&\geq 4\cdot \wgt(\cut_{j^*}(S))\geq U_{j^*}\geq \min_{j\in B_{i+1}}{U_j},
\end{align*}
implying that $\vb$ is a feasible solution to the semidefinite program used in the call of $\SDP(\game, S, B_{i+1})$. The second last inequality follows using Claim~\ref{claim:potential-increase-by-p-move}, since the $\tau$-approximate (and, consequently, the $3$-approximate) equilibrium condition is satisfied for player $j^*$.

Now, observe that edge $(j,k)$ belongs to $\cut_R(S)$ if it belongs to $\cut_{B_{i+1}}(S)$ and at least one of the players $j$ and $k$ belongs to $R$. Also, the quantity $\frac{1}{4}\cdot (3+v_j\cdot \widehat{v}+v_k\cdot \widehat{v}-v_j\cdot v_k)$ is equal to $1$ if one of the players $j$ and $k$ belongs to $R$ and is equal to $0$, otherwise. Thus,
\begin{align}\label{eq:phi_R-of-S}
\Phi_R(S) &=\sum_{(j,k)\in \cut_R(S)}{w_{jk}}= \frac{1}{4}\cdot \sum_{(j,k)\in \cut_{B_{i+1}}(S)}{w_{jk}\cdot (3+v_j\cdot \widehat{v}+v_k\cdot \widehat{v}-v_j\cdot v_k)}. 
\end{align}

Furthermore, observe that edge $(j,k)$ belongs to $\cut_R(S_{-R},S'_R)$ if it belongs to $\cut_{B_{i+1}}(S)$ and both players $j$ and $k$ belong to $R$, or it does not belong to $\cut_{B_{i+1}}(S)$ and exactly one of players $j$ and $k$ belongs to $R$. The quantity $\frac{1}{4}\cdot (1+v_j\cdot \widehat{v}+v_k\cdot \widehat{v}+v_j\cdot v_k)$ is equal to $1$ if both players $j$ and $k$ belong to $R$ and is equal to $0$, otherwise. Also, the quantity $\frac{1}{2}\cdot (1-v_j\cdot v_k)$ is equal to $1$ if exactly one of the players $j$ and $k$ belongs to $R$ and is equal to $0$, otherwise. Thus,
\begin{align}\nonumber
    \Phi_R(S_{-R},S'_R) &=\sum_{(j,k)\in \cut_R(S_{-R},S'_R)}{w_{jk}}\\\nonumber
    &=\frac{1}{2}\cdot \sum_{(j,k)\in E\setminus \cut_{B_{i+1}}(S)}{w_{jk}\cdot (1-v_j\cdot v_k)}\\\label{eq:phi_R-of-deviation-from-S}
    &\quad\,+\frac{1}{4}\cdot \sum_{(j,k)\in \cut_{B_{i+1}}(S)}{w_{jk}\cdot (1+v_j\cdot \widehat{v}+v_k\cdot \widehat{v}+v_j\cdot v_k)}.
\end{align}
Using equations (\ref{eq:phi_R-of-S}) and (\ref{eq:phi_R-of-deviation-from-S}) and the definition of the semidefinite program used in the call of $\SDP(\game,S,B_{i+1})$, our assumption (\ref{eq:assumption}) yields
\begin{align*}
    0 &< \Phi_R(S_{-R},S'_R) - \sigma\cdot \Phi_R(S)\\
    &=\frac{1}{2}\cdot \sum_{(j,k)\in E\setminus \cut_{B_{i+1}}(S)}{w_{jk}\cdot (1-v_j\cdot v_k)}\\
    &\quad\,-\frac{1}{4}\cdot \sum_{(j,k)\in \cut_{B_{i+1}}(S)}{w_{jk}\cdot  (3\sigma-1+(\sigma-1)\cdot v_j\cdot \widehat{v}+(\sigma-1)\cdot v_k\cdot\widehat{v}-(\sigma+1)\cdot v_j\cdot v_k)}\\
    &= Z,
\end{align*}
which contradicts the condition on $Z$. The lemma follows.
\end{proof}

We now aim to prove (in Lemma~\ref{lem:potential-increase-by-rounding} below) that, whenever the routine $\ROUND$() is executed on the set of vectors produced by a previous execution of routine $\SDP$(), the potential increases considerably. To do so, we will use an important property of each rounding iteration, which we state in Lemma~\ref{lem:xor}; its proof is deferred to Section~\ref{sec:xor-proofs}. We denote by $\xor(j,k)$ the event that the random hyperplane with unit normal vector $r$ separates the rotated vectors $v'_j$ and $v'_k$, i.e., the quantities $v'_j\cdot r$ and $v'_j\cdot r$ have different signs.
\begin{lemma}\label{lem:xor}
Let $v_j$ and $v_k$ be unit vectors in $\R^n$ and let $v'_j$ and $v'_k$ be the corresponding rotated vectors with respect to another unit vector $\widehat{v}\in \R^n$. Then,
\begin{align}\label{eq:xor-lower}
    \Pr[\xor(j,k)]\geq \frac{1}{4}\cdot(1-v_j\cdot v_k)
\end{align}
and 
\begin{align}\label{eq:xor-upper}
    \Pr[\xor(j,k)]\leq \frac{1}{8}\cdot(3\rho-1+(\rho-1)\cdot v_j\cdot \widehat{v}+(\rho-1)\cdot v_k\cdot \widehat{v}-(\rho+1)\cdot v_j\cdot v_k).
\end{align}
\end{lemma}

\begin{lemma}\label{lem:potential-increase-by-rounding}
Consider the execution of routine $\SDP(\game,S,B_{i+1})$ during phase $i\geq 0$. If the solution $\vb$ to the semidefinite program computed by $\SDP(\game,S,B_{i+1})$ has objective value $Z>0$, then, the probability that the subsequent call to $\ROUND(\game,S,B_{i+1},v)$ fails to identify a set of players $R\subseteq B_{i+1}$ so that $\Phi_R(S_{-R},S'_R)-\Phi_R(S)\geq \frac{\varepsilon}{48}\cdot W_{i+2}$ is at most $\frac{\varepsilon^4}{414720n^4}$. 
\end{lemma}

\begin{proof}
Consider the execution of routine $\ROUND$() in phase $i$ and a rounding iteration. Let $R\subseteq B_{i+1}$ be the set of players selected. Observe that the increase in the potential $\Phi_R$ from state $S$ to state $(S_{-R},S'_R)$ is equal to the total weight of the edges $(j,k)$ that do not belong to $\cut_{B_{i+1}}(S)$ and exactly one of the players $j$ and $k$ change strategy, minus the total weight of the edges $(j,k)$ that belong to $\cut_{B_{i+1}}(S)$ and exactly one of the players $j$ and $k$ change strategy. Using this observation and Lemma~\ref{lem:xor}, we have
\begin{align*}
&\E[\Phi_R(S_{-R},S'_R)-\Phi_R(S)]\\
&= \sum_{(j,k)\in E\setminus \cut_{B_{i+1}}(S)}{w_{jk}\cdot \Pr[\xor(j,k)]}-\sum_{(j,k)\in \cut_{B_{i+1}}(S)}{w_{jk}\cdot \Pr[\xor(j,k)]}\\
&\geq \frac{1}{4}\cdot \sum_{(j,k)\in E\setminus \cut_{B_{i+1}}(S)}{w_{jk}\cdot (1-v_j\cdot v_k)}\\
    &\quad\,-\frac{1}{8}\cdot \sum_{(j,k)\in \cut_{B_{i+1}}(S)}{w_{jk}\cdot (3\rho-1+(\rho-1)\cdot v_j\cdot \widehat{v}+(\rho-1)\cdot v_k\cdot\widehat{v}-(\rho+1)\cdot v_j\cdot v_k)}\\
    &=\frac{Z}{2}+\frac{\sigma-\rho}{8}\cdot \sum_{(j,k)\in \cut_{B_{i+1}}(S)}{w_{jk}\cdot (3+v_j\cdot \widehat{v}+v_k\cdot \widehat{v}-v_j\cdot v_k)}\\
    &> \frac{\varepsilon}{24}\cdot W_{i+2}.
\end{align*}
The last inequality follows by condition $Z>0$, the constraint of the semidefinite program and the definition of $W_{i+2}$, and the fact that $\sigma-\rho=\varepsilon/3$.

Now, observe that the random variable $\sum_{j\in B_{i+1}}{U_j}-\Phi_R(S_{-R},S'_R)+\Phi_R(S)$ is always non-negative. Applying the Markov inequality on it, we get
\begin{align*}
    &\Pr\left[\Phi_R(S_{-R},S'_R)-\Phi_R(S)\leq \frac{\varepsilon}{48}\cdot W_{i+2}\right]\\
    &=\Pr\left[\sum_{j\in B_{i+1}}{U_j}-\Phi_R(S_{-R},S'_R)+\Phi_R(S) \geq \sum_{j\in B_{i+1}}{U_j}-\frac{\varepsilon}{48}\cdot W_{i+2}\right]\\
    &\leq \frac{\sum_{j\in B_{i+1}}{U_j}-\E[\Phi_R(S_{-R},S'_R)-\Phi_R(S)]}{\sum_{j\in B_{i+1}}{U_j}-\frac{\varepsilon}{48}\cdot W_{i+2}}\\
    &< \frac{\sum_{j\in B_{i+1}}{U_j}-\frac{\varepsilon}{24}\cdot W_{i+2}}{\sum_{j\in B_{i+1}}{U_j}-\frac{\varepsilon}{48}\cdot W_{i+2}}\leq 1-\frac{\frac{\varepsilon}{48}\cdot W_{i+2}}{\sum_{j\in B_{i+1}}{U_j}}\leq 1-\frac{\varepsilon\cdot W_{i+2}}{48n\cdot W_{i+1}}=1-\frac{\varepsilon^3}{23040n^2}.
\end{align*}
I.e., the probability that $\Phi_R(S_{-R},S'_R)-\Phi_R(S)\leq \frac{\varepsilon}{48}\cdot W_{i+2}$ in a rounding iteration is at most $1-\frac{\varepsilon^3}{23040n^2}$. Thus, the probability that the potential increase is at most $\frac{\varepsilon}{48}\cdot W_{i+2}$ after all the  $\frac{23040n^2}{\varepsilon^3}\cdot \ln{\left(\frac{414720n^4}{\varepsilon^4}\right)}$ rounding iterations in an execution of $\ROUND$() is at most 
\begin{align*}
    \left(1-\frac{\varepsilon^3}{23040n^2}\right)^{\frac{23040n^2}{\varepsilon^3}\cdot \ln{\left(\frac{414720n^4}{\varepsilon^4}\right)}} \leq \frac{\varepsilon^4}{414720n^4}.
\end{align*}
The last inequality follows by the property $(1-1/t)^t\leq 1/e$.
\end{proof}

\section{Proof of Theorem~\ref{thm:main}}\label{sec:proof}
We now show how to use the main properties of global moves from Section~\ref{sec:global-moves}, as well as three additional claims (Claims~\ref{claim:small_block},~\ref{claim:stretch}, and~\ref{claim:p-move}) to prove Theorem~\ref{thm:main}. In our proof, we denote by $S^i$ the state reached at the end of phase $i\geq 0$, i.e., $\Sout=S^{m}$. We also denote by $R_i$ the set of players that changed their strategy during phase $i$. 

For every phase $i>0$, we define the following set of edges:
\begin{align*}
A_i &= \{(j,k)\in \cut_{R_i}(S^{i-1}): j,k\in R_i\cap B_i\}\\
C_i &= \{(j,k)\in \cut_{R_i}(S^{i-1}): j,k\in R_i\cap B_{i+1}\}\\
D_i &= \{(j,k)\in \cut_{R_i}(S^{i-1}): j\in R_i\cap B_{i}, k\in R_i\cap B_{i+1}\}\\
F_i &= \{(j,k)\not\in \cut_{R_i}(S^{i-1}): j\in R_i\cap B_{i}, k\in R_i\cap B_{i+1}\}\\
H_i &= \{(j,k)\in \cut_{R_i}(S^{i-1}): j\in R_i\cap B_{i}, k\not\in R_i\}\\
I_i &= \{(j,k)\not\in \cut_{R_i}(S^{i-1}): j\in R_i\cap B_{i}, k\not\in R_i\}\\
J_i &= \{(j,k)\in \cut_{R_i}(S^{i-1}): j\in R_i\cap B_{i+1}, k\not\in R_i\}\\
K_i &= \{(j,k)\not\in \cut_{R_i}(S^{i-1}): j\in R_i\cap B_{i+1}, k\not\in R_i\}
\end{align*}
We will express several properties maintained by our algorithm in terms of the weights of the above edge sets. The first one uses the definition of block $B_{i+1}$.

\begin{claim} \label{claim:small_block}
For every phase $i>0$, it holds $\wgt(C_i)+\wgt(D_i)+\wgt(F_i)+\wgt(J_i)+\wgt(K_i)\leq n\cdot W_{i+1}$.
\end{claim}

The next claim uses the main property guaranteed by the last call of routine $\SDP$() in phase $i-1$.

\begin{claim} \label{claim:stretch}
For every phase $i>0$, it holds $-(\sigma-1)\cdot \wgt(A_i) - \sigma\cdot \wgt(D_i) - \sigma\cdot \wgt(H_i) + \wgt(F_i) + \wgt(I_i)\leq 0$. 
\end{claim}

Finally, our third claim follows by the fact that the players in block $R_i\cap B_i$ make $\tau$-moves in phase $i$.

\begin{claim} \label{claim:p-move}
For every phase $i>0$, it holds $(\tau-1)\cdot \wgt(A_i)+(\tau-1)\cdot \wgt(D_i)+\tau\cdot \wgt(H_i)-\wgt(I_i)+\tau\cdot \wgt(J_i)-\tau\cdot \wgt(K_i) <0$.
\end{claim}

We are now ready to prove Theorem~\ref{thm:main}. We do so in Lemmas~\ref{lem:equilibrium} and~\ref{lem:running-time} below. Assuming that all calls of routine $\ROUND$() are successful in identifying a set if players whose change of strategies that increase the potential considerably, Lemma~\ref{lem:equilibrium} uses Claims~\ref{claim:small_block},~\ref{claim:stretch}, and~\ref{claim:p-move} to show that every players stays at a $(\rho+\varepsilon)$-equilibrium after the phase in which its strategy is irrevocably decided. 

\begin{lemma}\label{lem:equilibrium}
If no call to routine $\ROUND$() ever fails, the state $\Sout$ is a $(\rho+\varepsilon)$-approximate pure Nash equilibrium.
\end{lemma}

\begin{proof}
Let $j$ be a player of block $B_i$. Notice that, if $i=m$, player $j$ has no $\tau$-move (and, hence, no $(\rho+\varepsilon)$-move) to make after the execution of the last phase. So, in the following, we assume that $0<i\leq m-1$. 

Let $M_j$ be the set of edges in $\cut_j(S^i)$ whose other endpoint corresponds to a player in blocks $B_{i+1}$, ..., $B_{m-1}$. We will first bound $\wgt(M_j)$. Notice that $M_j\subseteq \cup_{t>i}{\left(H_t \cup J_t\right)}$. Hence,
\begin{align}\label{eq:H-J}
    \wgt(M_j) &\leq \sum_{t>i}{(\wgt(H_t)+\wgt(J_t)}).
\end{align}
Summing the inequalities in the statements of Claims~\ref{claim:small_block},~\ref{claim:stretch}, and~\ref{claim:p-move}, after multiplying the inequality in Claim~\ref{claim:small_block} by $\tau$, we get
\begin{align*}
    &(\tau-\sigma)\cdot \wgt(A_t)+\tau\cdot \wgt(C_t)+(2\tau-\sigma-1)\cdot \wgt(D_t)+(\tau+1)\cdot \wgt(F_t)\\
    &+(\tau-\sigma)\cdot \wgt(H_t)+2\tau\cdot \wgt(J_t) \leq  \tau\cdot n\cdot W_{t+1}.
\end{align*}
Notice that the LHS is lower-bounded by $\frac{\varepsilon}{3} \cdot (\wgt(H_t)+\wgt(J_t))$. Thus,
\begin{align*}
    \wgt(H_t)+\wgt(J_t) &\leq 9n\cdot W_{t+1}/\varepsilon.
\end{align*}
Now, inequality (\ref{eq:H-J}) yields
\begin{align}\label{eq:bound-on-M_j}
    \wgt(M_j) &\leq \frac{9n}{\varepsilon}\cdot \sum_{t>i}{W_{t+1}}\leq  \frac{9n}{\varepsilon}\cdot W_{i+2} \sum_{t=i+1}^\infty{\Delta^{i+1-t}}\leq 10n\cdot W_{i+2}/\varepsilon.
\end{align}

Now, notice that the edges in $M_j$ are those which can be removed from the cut due to moves by players in phases $i+1$, ..., $m$, and thus decrease the utility of player $j$ by up to $\wgt(M_j)$. I.e., 
\begin{align}\label{eq:decrease-in-utility}
u_j(S^{m})&\geq u_j(S^i)-\wgt(M_j).
\end{align}
Furthermore, the edges in $M_j$ are those which are not included in $\cut_j(S^i_{-j},s'_j)$ and thus do not contribute to the utility player $j$ would have by changing her strategy. Moves by players in phases $i+1$, ..., $m$ may increase the utility player $j$ would have by deviating by at most $\wgt(M_j)$, i.e., \begin{align}\label{eq:increase-in-utility-by-deviating}
    u_j(S^i_{-j},s'_j) &\geq u_j(S^{m}_{-j},s'_j)-\wgt(M_j).
\end{align}

We are ready to prove that player $j$ has no $(\rho+\varepsilon)$-move to make after the execution of the last phase of the algorithm, using inequalities (\ref{eq:bound-on-M_j}), (\ref{eq:decrease-in-utility}), and (\ref{eq:increase-in-utility-by-deviating}), and the fact that player $j$ had no $\tau$-move to make at the end of phase $i$.
We have
\begin{align}\nonumber
    (\rho+\varepsilon)\cdot u_j(S^{m}) &\geq (\rho+\varepsilon)\cdot u_j(S^i)-10(\rho+\varepsilon)n\cdot W_{i+2}/\varepsilon\\\nonumber
    &\geq \tau\cdot u_j(S^i)-30n\cdot  W_{i+2}/\varepsilon+(\rho+\varepsilon-\tau)\cdot u_j(S^i)\\\nonumber
    &\geq u_j(S^i_{-j},s'_j)-30n\cdot W_{i+2}/\varepsilon+\frac{\varepsilon}{3}\cdot u_j(S^i)\\\label{eq:equilibrium-eq1}
    &\geq u_j(S^{m}_{-j},s'_j)-40n\cdot W_{i+2}/\varepsilon+\frac{\varepsilon}{3}\cdot u_j(S^i).
\end{align}
Now, since player $j$ has no $\tau$-move at state $S^i$, Claim~\ref{claim:potential-increase-by-p-move} implies that $u_j(S^i)\geq W_{i+1}/4$ and, by the definition of the blocks, $u_j(S^i)\geq 120n\cdot W_{i+2}/\varepsilon^2$. Then, equation (\ref{eq:equilibrium-eq1}) yields $(\rho+\varepsilon)\cdot u_j(S^{m})\geq u_j(S^{m}_{-j},s'_j)$, as desired.
\end{proof}

Also, assuming that all calls of routine $\ROUND$() are successful, Lemma~\ref{lem:running-time} uses Claims~\ref{claim:small_block},~\ref{claim:stretch}, and~\ref{claim:p-move} to bound the total increase of the potential $\Phi_{R_i}$ in phase $i$ and Claim~\ref{claim:potential-increase-by-p-move} and Lemma~\ref{lem:potential-increase-by-rounding} to show that this increase happens after a small number of local and global moves.

\begin{lemma}\label{lem:running-time}
If no call to routine $\ROUND$() ever fails, the algorithm terminates after performing at most $\bigO(n^2/\varepsilon)$ $\tau$-moves, at most $\bigO(n^3/\varepsilon^4)$ calls to routine $\SDP$(), and at most $\frac{414720n^3}{\varepsilon^4}$ calls to routine $\ROUND$().
\end{lemma}

\begin{proof}
We will first bound the increase in the potential $\Phi_{R_i}$ of phase $i>0$, proving that
\begin{align}\label{eq:pot-increase}
\Phi_{R_i}(S^i)-\Phi_{R_i}(S^{i-1}) &\leq 18n \cdot W_{i+1}/\varepsilon.
\end{align}
Observe that $\cut_{R_i}(S^{i-1})=A_i\cup C_i \cup D_i \cup H_i\cup J_i$ and $\cut_{R_i}(S^{i})=A_i\cup C_i \cup D_i \cup I_i\cup K_i$. Hence, $\Phi_{R_i}(S^{i-1})=\wgt(A_i)+\wgt(C_i)+\wgt(D_i)+\wgt(H_i)+\wgt(J_i)$, $\Phi_{R_i}(S^i)=\wgt(A_i)+\wgt(C_i)+\wgt(D_i)+\wgt(I_i)+\wgt(K_i)$ and, thus,
\begin{align}\label{eq:pot-increase-eq1}
    \Phi_{R_i}(S^i)-\Phi_{R_i}(S^{i-1}) &= -\wgt(H_i)+\wgt(I_i)-\wgt(J_i)+ \wgt(K_i)
\end{align}
We now multiply the inequalities in Claims~\ref{claim:small_block},~\ref{claim:stretch}, and~\ref{claim:p-move} by $\frac{\sigma(\tau-1)}{\tau-\sigma}$, $\frac{\tau-1}{\tau-\sigma}$, and $\frac{\sigma-1}{\tau-\sigma}$, respectively, and sum them up. We obtain 
\begin{align*}
    &\frac{\sigma(\tau-1)}{\tau-\sigma}\cdot \wgt(C_i)+\frac{(\sigma-1)(\tau-1)}{\tau-\sigma}\cdot \wgt(D_i)+\frac{(\sigma+1)(\tau-1)}{\tau-\sigma}\cdot \wgt(F_i)\\
    &-\wgt(H_i)+\wgt(I_i)+\frac{2\tau\sigma-\tau-\sigma}{\tau-\sigma}\cdot \wgt(J_i)+\wgt(K_i) \leq  \frac{\sigma(\tau-1)}{\tau-\sigma}\cdot n\cdot W_{i+1}\leq 18n\cdot W_{i+1}/\varepsilon.
\end{align*}
The proof of inequality (\ref{eq:pot-increase}) follows by observing that all coefficients in the RHS of this last inequality besides that of term $\wgt(H_i)$ are positive and thus the RHS is lower-bounded by the LHS of (\ref{eq:pot-increase-eq1}).

By Claim~\ref{claim:potential-increase-by-p-move}, we have that a $\tau$-move increases the potential by at least $W_{i+1}/4$. By the bound on the potential increase in the whole phase $i$ from (\ref{eq:pot-increase}), we obtain that no more than $72n/\varepsilon$ $\tau$-moves take place in phase $i$. The bound on the total number of $\tau$-moves follows since there are at most $n$ phases containing some player.

Similarly, by Lemma~\ref{lem:potential-increase-by-rounding}, we have that a successful call to routine $\ROUND(\game,S,B_{i+1})$ at phase $i>0$ increases the potential by at least $\varepsilon \cdot W_{i+2}/48$. By the bound on the potential increase in the whole phase $i$ from (\ref{eq:pot-increase}) and the relation $W_{i+1}/W_{i+2}=\Delta=480n/\varepsilon^2$, we obtain that no more than $414720n^2/\varepsilon^4$ calls to routine $\ROUND(\game,S,B_{i+1})$ take place in phase $i$. Again, the bound on the total number of callls to $\ROUND$() follows since there are at most $n$ phases containing some player.

Regarding the number of calls to routine $\SDP()$ within a phase, we claim that this is at most the number of $\tau$-moves plus twice the number of calls to routine $\ROUND()$ plus $1$. Indeed, in each execution of the routine $\GM()$ in which $\ROUND()$ is executed $t>0$ times, there are at most $t+1\leq 2t$ executions of $\SDP()$. Also, in each execution of the routing $\GM()$ in which $\ROUND()$ is not executed, the number of executions of $\GM()$ (and, consequently, of $\SDP()$) is at most the number of $\tau$-moves, besides in the last round of the phase, where we do not have any $\tau$-move but there is one $\SDP()$ call. Hence, a bound of $\bigO(n^2/\varepsilon^4)$ on the number of $\SDP()$ calls in a phase follows
by the bounds on the number of $\tau$-moves and $\ROUND()$ calls in a phase. The bound of $\bigO(n^3/\varepsilon^4)$ on the total number of executions of $\SDP()$ follows since there are at most $n$ phases containing some player.
\end{proof}

To complete the proof of Theorem~\ref{thm:main} after having proved Lemmas~\ref{lem:equilibrium} and~\ref{lem:running-time}, it suffices to show that the probability that a call to routine $\ROUND$() fails to identify a set of players whose change of strategies increase the potential considerably is at most $1/n$. This can be done by a simple application of the union bound: there are at most $\frac{414720n^3}{\varepsilon^4}$ executions of routine $\ROUND$() and, by Lemma~\ref{lem:potential-increase-by-rounding}, each of them fails with probability at most $\frac{\varepsilon^4}{414720n^4}$.

\section{Properties of randomized rounding}\label{sec:xor-proofs}
In this section, we prove the two inequalities of Lemma~\ref{lem:xor}.
Let $\theta_j, \theta_k\in [0,\pi]$ be the angles vectors $v_j$ and $v_k$ form with vector $\widehat{v}$ and $\phi$ be the rotation angle vector $v_k$ forms with the two-dimensional plane defined by vectors $v_j$ and $\widehat{v}$. Without loss of generality, we can assume that $\theta_j\geq \theta_k$ and, furthermore, that vectors $\widehat{v}$, $v_j$, and $v_k$ are lying in three dimensions and are defined as $\widehat{v}=(1,0,0)$, $v_j=(\cos{\theta_j},\sin{\theta_j},0)$, and $v_k=(\cos{\theta_k},\sin{\theta_k}\cdot \cos{\phi},\sin{\theta_k}\cdot \sin{\phi})$. Denoting the rotation function by $f$ (i.e., $f(\theta)=\frac{\pi}{2}\cdot (1-\cos{\theta})$ for $\theta\in [0,\pi]$), we have that the rotated vectors of $v_j$ and $v_k$ are $v'_j=(\cos{f(\theta_j)},\sin{f(\theta_j)},0)$ and $v'_k=(\cos{f(\theta_k)},\sin{f(\theta_k)}\cdot \cos{\phi},\sin{f(\theta_k)}\cdot \sin{\phi})$, respectively. These definitions yield
\begin{align}\label{eq:prod-of-rotated-vecs}
    v'_j \cdot v'_k &= \cos{f(\theta_j)}\cdot \cos{f(\theta_k)}+\sin{f(\theta_j)}\cdot \sin{f(\theta_k)}\cdot \cos{\phi}
\end{align}
and
\begin{align}\label{eq:prod-of-vecs}
    v_j \cdot v_k &= \cos{\theta_j}\cdot \cos{\theta_k}+\sin{\theta_j}\cdot \sin{\theta_k}\cdot \cos{\phi}.
\end{align}
The probability that a random hyperplane separates the rotated vectors $v'_j$ and $v'_k$ is proportional to the angle between them, i.e., $\Pr[\xor(j,k)]=\frac{\arccos{(v'_j\cdot v'_k)}}{\pi}$. 

\subsection{Proof of inequality~(\ref{eq:xor-lower}) in Lemma~\ref{lem:xor}} 
To prove inequality (\ref{eq:xor-lower}), it suffices to prove that
\begin{align*}
    \frac{\arccos{(v'_j\cdot v'_k)}}{\pi} &\geq \frac{1-v_j\cdot v_k}{4},
\end{align*}
or, equivalently, 
\begin{align}\label{eq:main-ineq}
    v'_j\cdot v'_k - \cos{\left(\frac{\pi}{4}\cdot (1-v_j\cdot v_k)\right)} &\leq 0.
\end{align}
By (\ref{eq:prod-of-rotated-vecs}) and (\ref{eq:prod-of-vecs}), and by setting $g(x,y,\phi)=\cos{x}\cdot \cos{y}+\sin{x}\cdot \sin{y}\cdot \cos{\phi}$, the LHS of (\ref{eq:main-ineq}) becomes
\begin{align*}
    g(f(\theta_j),f(\theta_k),\phi)-\cos{\left(\frac{\pi}{4}\cdot (1-g(\theta_j,\theta_k,\phi))\right)}.
\end{align*}
Observe that $g$ is a linear function in $\cos{\phi}$, taking values in $[-1,1]$. To see why, observe that $g(x,y,\phi)\leq g(x,y,0)=\cos{(x-y)}$ and $g(x,y,\phi)\geq g(x,y,\pi)=\cos{(x+y)}$. Now the derivative of the LHS of (\ref{eq:main-ineq}) with respect to $\cos{\phi}$ is 
\begin{align*}
    \frac{\partial{g(f(\theta_j),f(\theta_k),\phi)}}{\partial{\cos{\phi}}}-\frac{\partial{g(\theta_j,\theta_k,\phi)}}{\partial{\cos{\phi}}}\cdot \frac{\pi}{4}\cdot \sin{\left(\frac{\pi}{4}\cdot (1-g(\theta_j,\theta_k,\phi))\right)},
\end{align*}
which is increasing in $\cos{\phi}$. To see why, observe that the derivative of function $g(x,y,\phi)$ with respect to $\cos{\phi}$ is equal to $\sin{x}\cdot \sin{y}$, which is non-negative when $x,y\in [0,\pi]$, and the quantity inside the $\sin$ is decreasing in $\cos{\phi}$, taking values between $0$ and $\pi/2$. We conclude that the LHS of (\ref{eq:main-ineq}) is convex in $\cos{\phi}$. Thus, it is maximized when either $\cos{\phi}=1$ or $\cos{\phi}=-1$ and it suffices to prove inequality (\ref{eq:main-ineq}) for $\cos{\phi}\in \{-1,1\}$. For $\cos{\phi}=1$, we have $g(x,y,\phi)=\cos{(x-y)}$, the LHS of (\ref{eq:main-ineq}) becomes
\begin{align*}
\cos{(f(\theta_j)-f(\theta_k))} - \cos{\left(\frac{\pi}{4}\cdot (1-\cos{(\theta_j-\theta_k)})\right)},
\end{align*}
and (\ref{eq:main-ineq}) is equivalent to
\begin{align}\label{eq:xor-lower-first-ineq}
    |f(\theta_j)-f(\theta_k)| - \frac{\pi}{4}\cdot (1-\cos{(\theta_j-\theta_k)})&\geq 0.
\end{align}
For $\cos{\phi}=-1$, we have $g(x,y,\phi)=\cos{(x+y)}$, the LHS of (\ref{eq:main-ineq}) becomes
\begin{align*}
\cos{(f(\theta_j)+f(\theta_k))} - \cos{\left(\frac{\pi}{4}\cdot (1-\cos{(\theta_j+\theta_k)})\right)},
\end{align*}
and (\ref{eq:main-ineq}) is equivalent to
\begin{align}\label{eq:xor-lower-second-ineq}
    f(\theta_j)+f(\theta_k) - \frac{\pi}{4}\cdot (1-\cos{(\theta_j+\theta_k)})&\geq 0.
\end{align}

We will complete the proof by showing that both (\ref{eq:xor-lower-first-ineq}) and (\ref{eq:xor-lower-second-ineq}) are true for the particular rotation function $f(\theta)=\frac{\pi}{2}\cdot (1-\cos{(\theta)})$ we use. Without loss of generality, assume that $0\leq \theta_k\leq \theta_j=\theta_k+x\leq \pi$ with $x\in [0,\pi-\theta_k]$. Then, using our specific rotation function, the LHS of inequality (\ref{eq:xor-lower-first-ineq}) can be written as
\begin{align*}
    \frac{\pi}{4}\cdot \left(2\cos{\theta_k}-2\cos{(\theta_k+x)}-1+\cos{x}\right) &= \frac{\pi}{2}\cdot \left(2\sin{\left(\theta_k+\frac{x}{2}\right)}\cdot \sin{\frac{x}{2}}-\sin^2{\frac{x}{2}}\right)\\
    &\geq \frac{\pi}{2}\cdot \sin^2{\frac{x}{2}}\geq 0,
\end{align*}
which proves inequality (\ref{eq:xor-lower-first-ineq}). The first inequality follows since our assumption for $\theta_k$ and $x$ imply that $\frac{x}{2}\leq \theta_k+\frac{x}{2}\leq \pi-\frac{x}{2}$, which yields $\sin{\left(\theta_k+\frac{x}{2}\right)}\geq \sin{\frac{x}{2}}$.

The LHS of inequality (\ref{eq:xor-lower-second-ineq}) becomes
\begin{align*}
    &\frac{\pi}{2}\cdot \left(2-\cos{(\theta_k+x)}-\cos{\theta_k}\right) - \frac{\pi}{4}\cdot \left(1-\cos{(2\theta_k+x)}\right)\\ &=\frac{\pi}{2}\cdot \left(1-2\cos{\left(\theta_k+\frac{x}{2}\right)}\cdot \cos{\frac{x}{2}}+\cos^2{\left(\theta_k+\frac{x}{2}\right)}\right).
\end{align*}
Since $\cos{\frac{x}{2}}$ takes values in $[-1,1]$, the quantity inside this last parenthesis is between the values $\left(1-\cos{\left(\theta_k+\frac{x}{2}\right)}\right)^2$ and $\left(1+\cos{\left(\theta_k+\frac{x}{2}\right)}\right)^2$, which are both clearly non-negative. The proof of inequality (\ref{eq:xor-lower-second-ineq}) follows. \qed

\subsection{Proof of inequality~(\ref{eq:xor-upper}) in Lemma~\ref{lem:xor}}
To prove inequality (\ref{eq:xor-upper}), it suffices to prove that
\begin{align}\label{eq:xor-upper-to-prove}
    \frac{\arccos{(v'_j\cdot v'_k)}}{\pi}-\frac{1}{8}\cdot \left(3\rho-1+(\rho-1)\cdot v_j\cdot\widehat{v}+(\rho-1)\cdot v_k\cdot \widehat{v}-(\rho+1)\cdot v_j\cdot v_k\right) \leq 0.
\end{align}
We have verified (\ref{eq:xor-upper-to-prove}) in its general form by extensive numerical computations using Mathematica (version 13.0) and explain how we do so later in this section. Until then, we present a formal analysis for the simpler case where the vectors $v_j$, $v_k$, and $\widehat{v}$ lie on the same $2$-dimensional plane. This includes the values of parameters $\theta_j$, $\theta_k$, and $\phi$ for which inequality (\ref{eq:xor-upper-to-prove}) is tight and has allowed us to compute the approximation factor $\rho$ explicitly, as a function of a quadratic equation. We distinguish between two cases, depending on whether $\phi$ is equal to $0$ or $\pi$.

\paragraph{Formal proof of inequality (\ref{eq:xor-upper-to-prove}) when $\phi=0$.} In this case, equations (\ref{eq:prod-of-rotated-vecs}) and (\ref{eq:prod-of-vecs}) yield 
\begin{align*}
    v'_j\cdot v'_k &= \cos{(f(\theta_j)-f(\theta_k))}=\cos{\left(\frac{\pi}{2}\cdot (\cos{\theta_k}-\cos{\theta_j})\right)}.
\end{align*}
and $v_j\cdot v_k=\cos{(\theta_j-\theta_k)}$, respectively, and the LHS of (\ref{eq:xor-upper-to-prove}) becomes
\begin{align}\nonumber
    &\frac{1}{2}\cdot (\cos{\theta_k}-\cos{\theta_j})-\frac{3\rho-1}{8}-\frac{\rho-1}{8}\cdot \left(\cos{\theta_j}+\cos{\theta_k}\right)+\frac{\rho+1}{8}\cdot \cos{(\theta_j-\theta_k)}\\\nonumber
    &= -\frac{\rho}{2}+\sin{\left(\frac{\theta_j-\theta_k}{2}\right)}\cdot \sin{\left(\frac{\theta_j+\theta_k}{2}\right)}-\frac{\rho-1}{4}\cdot \cos{\left(\frac{\theta_j-\theta_k}{2}\right)}\cdot \cos{\left(\frac{\theta_j+\theta_k}{2}\right)}\\\label{eq:xor-upper-phi=0-last-step}
    &\quad\, +\frac{\rho+1}{4}\cdot \cos^2{\left(\frac{\theta_j-\theta_k}{2}\right)}.
\end{align}
Applying the inequality $x\cdot y \leq c\cdot x^2+\frac{1}{4c}\cdot y^2$ for $c=\frac{1+\sqrt{13}}{4}$, $x=\sin{\left(\frac{\theta_j-\theta_k}{2}\right)}$, and $y=\sin{\left(\frac{\theta_j+\theta_k}{2}\right)}$, we obtain 
\begin{align*}\nonumber
    &\sin{\left(\frac{\theta_j-\theta_k}{2}\right)}\cdot \sin{\left(\frac{\theta_j+\theta_k}{2}\right)} \leq \frac{1+\sqrt{13}}{4}\cdot \sin^2{\left(\frac{\theta_j-\theta_k}{2}\right)}+\frac{\sqrt{13}-1}{12}\cdot \sin^2{\left(\frac{\theta_j+\theta_k}{2}\right)}\\
    &= \frac{1+2\sqrt{13}}{6}-\frac{1+\sqrt{13}}{4}\cdot \cos^2{\left(\frac{\theta_j-\theta_k}{2}\right)}-\frac{\sqrt{13}-1}{12}\cdot \cos^2{\left(\frac{\theta_j+\theta_k}{2}\right)}.
\end{align*}
Using this inequality and the fact that $\rho=\frac{1+2\sqrt{13}}{3}$, the RHS of equation (\ref{eq:xor-upper-phi=0-last-step}) is at most
\begin{align*}
    &-\frac{\rho}{2}+\frac{1+2\sqrt{13}}{6}-\frac{\sqrt{13}-1}{12}\cdot \cos^2{\left(\frac{\theta_j+\theta_k}{2}\right)}-\frac{\rho-1}{4}\cdot \cos{\left(\frac{\theta_j+\theta_k}{2}\right)}\cdot\cos{\left(\frac{\theta_j-\theta_k}{2}\right)}\\
    &-\left(\frac{1+\sqrt{13}}{4}-\frac{\rho+1}{4}\right)\cdot \cos^2{\left(\frac{\theta_j-\theta_k}{2}\right)} = -\frac{\sqrt{13}-1}{12}\cdot \left(\cos{\left(\frac{\theta_j+\theta_k}{2}\right)}+\cos{\left(\frac{\theta_j-\theta_k}{2}\right)}\right)^2,
\end{align*}
which is clearly non-positive, completing the proof.

We remark that this last RHS expression is equal to $0$ only when $\theta_j=\pi$. In this case, the RHS of (\ref{eq:xor-upper-phi=0-last-step}) becomes 
\begin{align*}
    -\frac{\rho}{2}+\cos^2{\left(\frac{\theta_k}{2}\right)}+\frac{\rho-1}{4}\cdot \sin^2{\left(\frac{\theta_k}{2}\right)}+\frac{\rho+1}{4}\cdot \sin^2{\left(\frac{\theta_k}{2}\right)}=\frac{2-\rho}{2}\cdot \left(1-\sin^2{\left(\frac{\theta_k}{2}\right)}\right),
\end{align*}
implying that $(\pi,\pi)$ is the only pair of values for $\theta_j$ and $\theta_k$ for which inequality (\ref{eq:xor-upper-to-prove}) is tight.

\paragraph{Formal proof of inequality (\ref{eq:xor-upper-to-prove}) when $\phi=\pi$.} Then, equations (\ref{eq:prod-of-rotated-vecs}) and (\ref{eq:prod-of-vecs}) yield 
\begin{align*}
    v'_j\cdot v'_k &= \cos{(f(\theta_j)+f(\theta_k))}=\cos{\left(\pi-\frac{\pi}{2}\cdot (\cos{\theta_j}+\cos{\theta_k})\right)}.
\end{align*}
and $v_j\cdot v_k=\cos{(\theta_j+\theta_k)}$, respectively. Notice that 
\begin{align*}
    \frac{\arccos{(v'_j\cdot v'_k)}}{\pi} &=1-\frac{1}{2}\cdot |\cos{\theta_j}+\cos{\theta_k}|.
\end{align*}
So, we distinguish between two subcases, depending on whether $\cos{\theta_j}+\cos{\theta_k}$ is positive or not (and, consequently, on whether $\theta_j+\theta_k$ is lower than $\pi$ or not).

If $\cos{\theta_j}+\cos{\theta_k}> 0$ (and $\theta_j+\theta_k<\pi$), then $\frac{\arccos{(v'_j\cdot v'_k)}}{\pi}=1-\frac{1}{2}\cdot (\cos{\theta_j}+\cos{\theta_k})$ and the LHS of inequality (\ref{eq:xor-upper-to-prove}) becomes
\begin{align*}
    &\frac{9-3\rho}{8}-\frac{\rho+3}{8}\cdot (\cos{\theta_j}+\cos{\theta_k})+\frac{\rho+1}{8}\cdot \cos{(\theta_j+\theta_k)}\\
    &= \frac{2-\rho}{2}-\frac{\rho+3}{4}\cdot \cos{\left(\frac{\theta_j-\theta_k}{2}\right)}\cdot \cos{\left(\frac{\theta_j+\theta_k}{2}\right)}+\frac{\rho+1}{4}\cos^2{\left(\frac{\theta_j+\theta_k}{2}\right) }\\
    &\leq \frac{2-\rho}{2}-\frac{1}{2}\cdot \cos^2{\left(\frac{\theta_j+\theta_k}{2}\right)} < 0.
\end{align*}
The first inequality follows since $\cos{\left(\frac{\theta_j+\theta_k}{2}\right)}> 0$ and $\theta_k\geq 0$ which implies  $\cos{\left(\frac{\theta_j-\theta_k}{2}\right)}\geq  \cos{\left(\frac{\theta_j+\theta_k}{2}\right)}$. 

Otherwise, if $\cos{\theta_j}+\cos{\theta_k}\leq 0$ (and $\theta_j+\theta_k\geq \pi$), then $\frac{\arccos{(v'_j\cdot v'_k)}}{\pi}=1+\frac{1}{2}\cdot (\cos{\theta_j}+\cos{\theta_k})$ and the LHS of inequality (\ref{eq:xor-upper-to-prove}) becomes
\begin{align*}
    &\frac{9-3\rho}{8}+\frac{5-\rho}{8}\cdot (\cos{\theta_j}+\cos{\theta_k})+\frac{\rho+1}{8}\cdot \cos{(\theta_j+\theta_k)}\\
    &= \frac{2-\rho}{2}+\frac{5-\rho}{4}\cdot \cos{\left(\frac{\theta_j-\theta_k}{2}\right)}\cdot \cos{\left(\frac{\theta_j+\theta_k}{2}\right)}+\frac{\rho+1}{4}\cos^2{\left(\frac{\theta_j+\theta_k}{2}\right) }\\
    &\leq \frac{2-\rho}{2}+\frac{\rho-2}{2}\cdot \cos^2{\left(\frac{\theta_j+\theta_k}{2}\right)} \leq 0.
\end{align*}
The first inequality follows since $\cos{\left(\frac{\theta_j+\theta_k}{2}\right)}\leq 0$ and $\cos{\left(\frac{\theta_j-\theta_k}{2}\right)}\geq - \cos{\left(\frac{\theta_j+\theta_k}{2}\right)}$. To see why this last inequality is true, observe that the assumption $\theta_j\in [0,\pi]$ implies that $\frac{\theta_j-\theta_k}{2}\leq \pi-\frac{\theta_j+\theta_k}{2}$ and, hence, $\cos{\left(\frac{\theta_j-\theta_k}{2}\right)}\geq \cos{\left(\pi-\frac{\theta_j+\theta_k}{2}\right)}=-\cos{\left(\frac{\theta_j+\theta_k}{2}\right)}$.

Notice that, again, the last inequality is tight only for the pair of values $(\pi,\pi)$ for $\theta_j$ and $\theta_k$.

\paragraph{Verifying inequality (\ref{eq:xor-upper-to-prove}) numerically.} We now consider the general case. Define the function
\begin{align*}
    \lambda(x,y,z) &= \frac{\arccos{\left(\cos{f(x)}\cdot \cos{f(y)}+z\cdot \sin{f(x)}\cdot \sin{f(y)}\right)}}{\pi}-\frac{3\rho-1}{8}-\frac{\rho-1}{8}\cdot \cos{x}\\
    &\quad-\frac{\rho-1}{8}\cdot \cos{y}+\frac{\rho+1}{8}\cdot \cos{x}\cdot \cos{y}+\frac{\rho+1}{8}\cdot z \cdot \sin{x}\cdot \sin{y},
\end{align*}
for $x\in [0,\pi]$, $y\in [0,\pi]$, and $z\in [L,H]$ with $L=\frac{-\cos{f(x)}\cos{f(y)}-1}{\sin{f(x)}\sin{f(y)}}$ and $H=\frac{-\cos{f(x)}\cos{f(y)}+1}{\sin{f(x)}\sin{f(y)}}$. Notice that $[-1,1]$ is a subinterval of $[L,H]$. Now, observe that $\lambda(\theta_j,\theta_k,\cos{\phi})$ is identical to the LHS of inequality (\ref{eq:xor-upper-to-prove}), after substituting $v'_j\cdot v'_k$ and $v_j\cdot v_k$ using (\ref{eq:prod-of-rotated-vecs}) and (\ref{eq:prod-of-vecs}), respectively. The derivative of $\lambda(x,y,z)$ with respect to $z$ is
\begin{align*}
    \frac{\partial{\lambda(x,y,z)}}{\partial{z}} &= -\frac{\sin{f(x)}\cdot \sin{f(y})}{\pi\cdot \sqrt{1-\left(\cos{f(x)}\cdot \cos{f(y)}+z\cdot \sin{f(x)}\cdot \sin{f(y)}\right)^2}}+\frac{\rho+1}{8}\cdot \sin{x}\cdot \sin{y},
\end{align*}
and has two roots:
\begin{align*}
    z_1(x,y) &= -\frac{\cos{f(x)}\cos{f(y)}}{\sin{f(x)}\sin{f(y)}}-\frac{1}{\sin{f(x)}\cdot \sin{f(y)}}\cdot \sqrt{1-\left(\frac{8\sin{f(x)}\sin{f(y)}}{\pi(\rho+1)\cdot \sin{x}\sin{y}}\right)^2}
\end{align*}
and 
\begin{align*}
    z_2(x,y) &= -\frac{\cos{f(x)}\cos{f(y)}}{\sin{f(x)}\sin{f(y)}}+\frac{1}{\sin{f(x)}\cdot \sin{f(y)}}\cdot \sqrt{1-\left(\frac{8\sin{f(x)}\sin{f(y)}}{\pi(\rho+1)\cdot \sin{x}\sin{y}}\right)^2}.
\end{align*}
By inspecting the derivative, we can see that the function $\lambda(x,y,z)$ is decreasing as $z$ takes values in $[L,z_1(x,y)]$, has a local minimum at $z_1(x,y)$, increases as $z$ takes values in $[z_1(x,y),z_2(x,y)]$, has a local maximum at $z_2(x,y)$, and decreases in $[z_2(x,y),H]$. To prove inequality (\ref{eq:xor-upper-to-prove}), it suffices to show that
\begin{align}\label{eq:to-verify-numerically}
\max_{\theta_j,\theta_k,\phi\in [0,\pi]}{\lambda(\theta_j,\theta_k,\cos{\phi}})\leq 0.
\end{align}
Given the formal bounds on $\lambda(\theta_j,\theta_k,1)$ and $\lambda(\theta_j,\theta_k,-1)$ in our analysis above, we complete the proof of (\ref{eq:to-verify-numerically}) by showing through extensive numerical computations using Mathematica that the quantity
\begin{align}\label{eq:nullify-deriv}
\max_{\substack{\theta_j,\theta_k\in [0,\pi]:\\z_2(\theta_j,\theta_k)\in [-1,1]}}{\lambda(\theta_j,\theta_k,z_2(\theta_j,\theta_k))}
\end{align}
is upper-bounded by $0$. The contour plot at the left of Figure~\ref{fig:plots} shows the values of the quantity $\max_{\phi\in [0,\pi]}{\lambda(\theta_j,\theta_k,\cos{\phi})}$ for all pairs $(\theta_j,\theta_k)\in [0,\pi]^2$. The values in the area defined by the dashed curves are those computed using (\ref{eq:nullify-deriv}); the remaining values are the maximum between $\lambda(\theta_j,\theta_k,1)$ and $\lambda(\theta_j,\theta_k,-1)$. In this way, we verify that (\ref{eq:to-verify-numerically}) is true with the inequality being strict everywhere besides the single point $(\pi,\pi)$. Recall that the formal analysis above has verified that the pair $(\pi,\pi)$ is the only point at which inequality (\ref{eq:xor-upper-to-prove}) is tight when vectors $v_j$, $v_k$, and $\widehat{v}$ lie on the same $2$-dimensional plane. The plot at the right of Figure~\ref{fig:plots} depicts the value of $\lambda(\theta_j,\theta_k,z_2(\theta_j,\theta_k))$ for $\theta_j=\theta_k=\theta$ being very close to $\pi$. 
\begin{figure}[ht]
     \centering
     \begin{subfigure}[b]{0.55\textwidth}
         \centering
         \includegraphics[width=\textwidth]{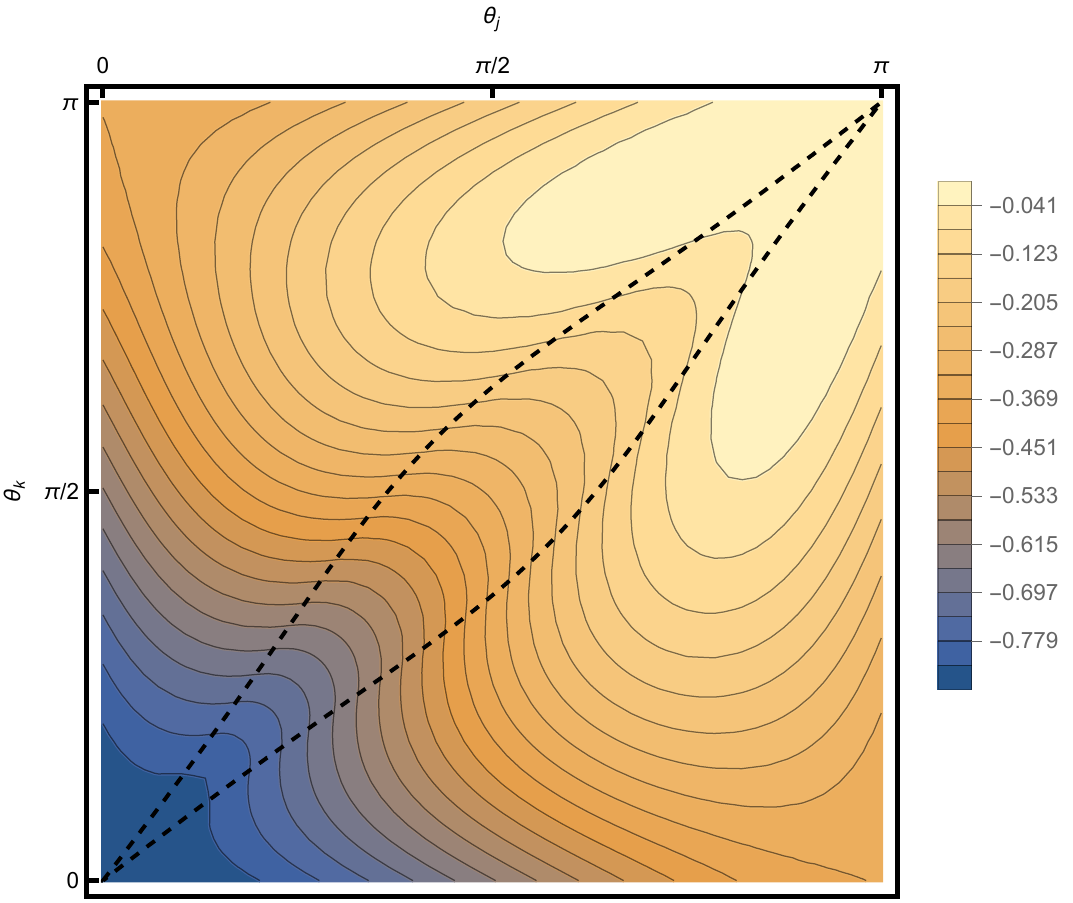}
          \end{subfigure}
     \begin{subfigure}[b]{0.44\textwidth}
         \centering
         \includegraphics[width=\textwidth]{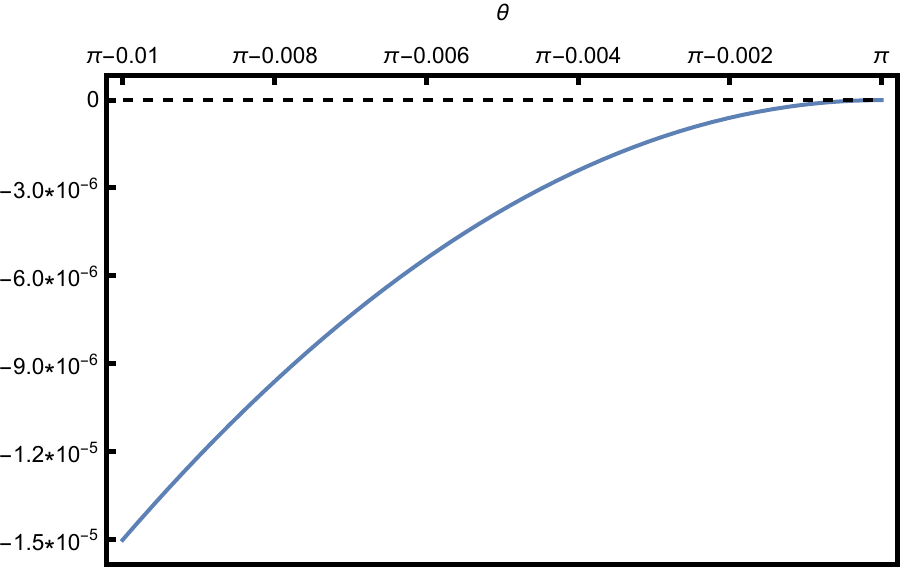}
          \end{subfigure}
    \caption{Numerical evaluation of the quantity $\max_{\phi\in [0,\pi]}{\lambda(\theta_j,\theta_k,\cos{\phi}})$ for all pairs $(\theta_j,\theta_k)\in [0,\pi]^2$ (left) and pairs $(\theta_j,\theta_k)$ with $\theta_j=\theta_k=\theta\in [0,\pi]$ (right).}
    \label{fig:plots}
\end{figure}

\section{Extensions and open problems}\label{sec:open}
We have shown how to compute $2.7371$-approximate pure Nash equilibria in cut games in polynomial time, providing the first improvement to the $3$-approximation bound of Bhalgat et al.~\cite{BCK10}. With a few minor modifications, our algorithm works for party affiliation games as well. These are generalizations of cut games, in which an edge can be either {\em enemy} or {\em friend}, giving utility to the players controlling its endpoints when they lie on different sides or the same side of the cut, respectively. 

As future work, it would be interesting to explore whether other rotation functions, possibly combined with more sophisticated rounding techniques (like those discussed in~\cite{LLZ02}), can give further improved results. How close to almost exact equilibria can we go? Even though we have make progress on this question here, the gap is still large. Exploring the inapproximability of equilibria for cut games (possibly, by strengthening the PLS-hardness results for local MAX-CUT in bounded-degree graphs; see, e.g.,~\cite{ET11}) is another problem that deserves investigation.

\bibliographystyle{plain}
\bibliography{pne}

\newpage\appendix
\section{Omitted proofs}
\paragraph{Proof of Claim~\ref{claim:subgame-potential}.}
Consider a subgame of $\game$ defined by a subset $R\subseteq \N$ of players and let $S$ and $(S_{-j},s'_j)$ be two states of the subgame differing in the strategies of player $j\in R$. By definition, we have
\begin{align*}
    \Phi_R(S)-\Phi_R(S_{-j},s'_j) &=\wgt(\cut_R(S))-\wgt(\cut_R(S_{-j},s'_j))\\
    &=\wgt(\cut_j(S))-\wgt(\cut_j(S_{-j},s'_j))=u_j(S)-u_j(S_{-j},s'_j),
\end{align*}
as desired.
\qed

\paragraph{Proof of Claim~\ref{claim:potential-increase-by-p-move}.}
Assume that player $j$ performs a $p$-move, reaching the state $(S_{-j},s'_j)$ from $S$, after changing her strategy from $s_j$ to $s'_j$. Then, the increase in player $j$'s utility during the $p$-move is
\begin{align*}
    u_j(S_{-j},s'_j)-u_j(S) &> (p-1)\cdot u_j(S) = \frac{p-1}{2} \cdot u_j(S) +\frac{p-1}{2}\cdot (U_j-u_j(S_{-j},s'_j))\\ 
    &= \frac{p-1}{2} \cdot U_j - \frac{p-1}{2}\cdot (u_j(S_{-j},s'_j)-u_j(S)),
\end{align*}
i.e., $u_j(S_{-j},s'_j)-u_j(S) \geq \frac{p-1}{p+1}\cdot U_j$. Since, by Claim~\ref{claim:subgame-potential}, $\Phi_R$ is an exact potential for any subset of players $R$ that contains player $j$, we have that $\Phi_R(S_{-j},s'_j)-\Phi_R(S)\geq \frac{p-1}{p+1}\cdot U_j$ as well.

Now assume that player $j$ has no $p$-move to make at state $S$. We have
\begin{align*}
    U_j &= u_j(S)+u_j(S_{-j},s'_j) \leq (p+1)\cdot u_j(S).
\end{align*}
The claim follows. \qed

\paragraph{Proof of Claim~\ref{claim:small_block}.}
Indeed, observe that the edges in sets $C_i$, $D_i$, $F_i$, $J_i$, and $K_i$ have at least one of their endpoint corresponding to a player in block $B_{i+1}$. The inequality follows since the total weight in edges incident to (the at most $n$) nodes corresponding to players in block $B_{i+1}$ is at most $W_{i+1}$.
\qed

\paragraph{Proof of Claim~\ref{claim:stretch}.}
By the definition of routine \GM(), in the last execution of routine \SDP() in phase $i-1$, the objective value of the semidefinite program returned is non-positive. Then, by applying Lemma~\ref{lem:sdp-stretch} for the set of players $R_i\cap B_i$, we get
\begin{align}\label{eq:stretch-eq1}
     \Phi_{R_i\cap B_i}(S^{i-1}_{-R_i\cap B_i}, S^i_{R_i\cap B_i}) &\leq \sigma \cdot \Phi_{R_i\cap B_i}(S^{i-1}).
\end{align}
By the definition of the edge sets, notice that $\cut_{R_i\cap B_i}(S^{i-1})=A_i\cup D_i\cup H_i$ and, hence, 
\begin{align}\label{eq:stretch-eq2}
\Phi_{R_i\cap B_i}(S^{i-1}) &=\wgt(A_i)+\wgt(D_i)+\wgt(H_i).
\end{align}
Furthermore, $\cut_{R_i\cap B_i}(S^{i-1}_{-R_i\cap B_i},S^i_{R_i\cap B_i})=A_i\cup F_i\cup I_i$ and, hence, \begin{align}\label{eq:stretch-eq3}
\Phi_{R_i\cap B_i}(S^{i-1}_{-R_i\cap B_i}, S^i_{R_i\cap B_i})&=\wgt(A_i)+\wgt(F_i)+\wgt(I_i).
\end{align}
The claim follows after using equations (\ref{eq:stretch-eq2}) and (\ref{eq:stretch-eq3}) to substitute the potential values in (\ref{eq:stretch-eq1}), and rearranging.
\qed

\paragraph{Proof of Claim~\ref{claim:p-move}.}
Consider player $j\in R_i\cap B_i$ and let $L_j$ be the set of edges incident to node $j$ that belong to the cut after the last $\tau$-move of player $j$ in phase $i$. The increase in the utility of player $j$ after this $\tau$-move is higher than $(1-1/\tau)\cdot \wgt(L_j)$. Now, since $\Phi_{R_i}$ is an exact potential (by Claim~\ref{claim:subgame-potential}), the increase of the potential value due to the last $\tau$-move of player $j$ is equal to the increase of her utility. Thus, the total increase of the potential between states $S^{i-1}$ and $S^i$ is higher than the total utility increase in the last $\tau$-moves of the players in $R_i\cap B_i$, i.e.,
\begin{align}\label{eq:p-move-eq1}
    \Phi_{R_i}(S^i) - \Phi_{R_i}(S^{i-1}) &> \sum_{j\in R_i\cap B_i}{(1-1/\tau)\cdot \wgt(L_j)} \geq (1-1/\tau)\cdot \wgt\left(\cup_{j\in R_i\cap B_i}{L_j}\right).
\end{align}
Now, notice that $\cut_{R_i}(S^{i})=A_i\cup C_i\cup D_i\cup I_i\cup K_i$, $\cut_{R_i}(S^{i-1})=A_i\cup C_i\cup D_i\cup H_i\cup J_i$, and $\cup_{j\in R_i\cap B_i}{L_j}\supseteq \cut_{R_i\cap B_i}(S^i)=A_i\cup D_i \cup I_i$. Hence,
\begin{align}\label{eq:p-move-eq2}
    \Phi_{R_i}(S^i) &= \wgt(A_i)+\wgt(C_i)+\wgt(D_i)+\wgt(I_i)+\wgt(K_i),\\\label{eq:p-move-eq3}
    \Phi_{R_i}(S^{i-1}) &= \wgt(A_i)+\wgt(C_i)+\wgt(D_i)+\wgt(H_i)+\wgt(J_i),
\end{align} 
and
\begin{align}\label{eq:p-move-eq4}
    \wgt\left(\cup_{j\in R_i\cap B_i}{L_j}\right) &\geq \wgt(A_i)+\wgt(D_i)+\wgt(I_i).
\end{align}
Using equations (\ref{eq:p-move-eq2}), (\ref{eq:p-move-eq3}), and (\ref{eq:p-move-eq4}) to substitute the potential values and $\wgt\left(\cup_{j\in R_i\cap B_i}{L_j}\right)$ in (\ref{eq:p-move-eq1}), we get
\begin{align*}
    \wgt(I_i)+\wgt(K_i)-\wgt(H_i)-\wgt(J_i) &> (1-1/\tau)\cdot (\wgt(A_i)+\wgt(D_i)+\wgt(I_i)).
\end{align*}
The claim follows by multiplying both sides of this last inequality by $\tau$ and rearranging.
\qed

\end{document}